\documentclass[a4paper,12pt]{article}%D
\usepackage{graphicx}
\usepackage{mathtools,amsfonts,amssymb,amscd,bm,amsthm,hyperref,verbatim}
\newcommand{\R}{{\mathbb R}}
\newcommand{\Z}{{\mathbb Z}}

\newcommand{\ep}{{\varepsilon}}

\newcommand{\be}{{\beta}}
\newcommand{\de}{{\delta}}

\newcommand{\De}{{\Delta}}

\newcommand{\si}{\sigma}
\newcommand{\vl}{{\; | \;}}

\newcommand{\id}{{\mathrm{id}}}
\newcommand{\Reeb}{\mathrm{Reeb}}
\newcommand{\ER}{\mathrm{EReeb}}
\newcommand{\BR}{\mathrm{BReeb}}
\newcommand{\core}{\mathrm{core}}
\newcommand{\PN}{\mathrm{PN}}

\newtheorem{thm}{Theorem}[section]
\newtheorem{prm}[thm]{Problem}
\newtheorem{lem}[thm]{Lemma}
\newtheorem{pro}[thm]{Proposition}
\newtheorem{cor}[thm]{Corollary}

\theoremstyle{definition}
\newtheorem{dfn}[thm]{Definition}

\begin{document}%D
%\pagewiselinenumbers
%\pagestyle{headings}

% Title of document, usually lower case except for first word
% and proper nouns.  Avoid unnecessary symbols.
\title{Book embeddings of Reeb graphs}

\author{Vitaliy Kurlin \\
 \url{http://kurlin.org} \\
 vitaliy.kurlin@durham.ac.uk \\
 Department of Mathematical Sciences  \\ 
 Durham University,  Durham DH1 3LE, UK
 }
% If needed, use a \thanks command, but not inside the \author
% command.
% \thanks{The first author was supported in part by a grant.}
% AMS 2000 Mathematics Subject Classification.  List one or several,
% separated by commas, ending in a period.
%\classification{57M05, 20F36, 05C25.}
% Keywords of the article, usually singular, no leading caps.  
% Separated by commas, ending with period.
%\keywords{graph, braid group, configuration space, fundamental group, homotopy type, deformation retraction, collision free motion, planning algorithm, complexity, robotics.}

% Abstract comes before maketitle
\maketitle 

\begin{abstract}
Let $X$ be a simplicial complex
 with a piecewise linear function $f:X\to\R$.
The \emph{Reeb graph} $\Reeb(f,X)$ is the quotient of $X$, 
 where we collapse each connected component of 
 $f^{-1}(t)$ to a single point.
Let the nodes of $\Reeb(f,X)$ be all homologically critical points where any homology of the corresponding component
 of the level set $f^{-1}(t)$ changes.
Then we can label every arc of $\Reeb(f,X)$ with the Betti numbers $(\beta_1,\beta_2,\dots,\beta_d)$ of the corresponding 
 $d$-dimensional component of a level set.
The homology labels give more information about 
 the original complex $X$ than the classical Reeb graph. 
We describe a canonical embedding of a Reeb graph
 into a multi-page book (a star cross a line) and
 give a unique linear code of this book embedding.
\end{abstract}
%\newpage

%1================================================
\section{Our contributions and related work}
\label{sec:results}

Reeb graphs are used in image analysis \cite{BGSF08} as simple representations of shapes.
The Reeb graph is an abstract graph that describes the linking structure of connected components in level sets of a real-valued function on a complex. 
Briefly, a complex is a union of vertices, edges, triangles, tetrahedra and so on.
%We can embed a graph if we can draw it without self-intersections.
Formal Definition~\ref{dfn:Reeb-graph} of a Reeb graph in section~\ref{sec:complexes-Reeb} doesn't specify any canonical way to visualize this abstract graph in a low dimensional space.

\begin{prm}
\label{prm:embedding}
Is it possible to canonically embed (draw without any self-intersections) and encode any general Reeb graph in a low-dimensional space? 
\end{prm}

Here we briefly introduce auxiliary concepts, see all definitions in sections~\ref{sec:complexes-Reeb}-\ref{sec:book-embeddings}. 
A cycle $C$ in a graph is independent of other cycles $C_1,\dots,C_k$ if (briefly) $C\subset\cup_{i=1}^k C_i$ and every edge of $(\cup_{i=1}^k C_i)-C$ is covered twice.
For instance, the theta-graph $\theta$ has two independent cycles (small loops) whose union covers the large cycle once and the middle arc twice. 
In general, $m$ multiple arcs between the same nodes generate $m-1$ independent cycles.
\smallskip

Embedding Theorem~\ref{thm:embedding} requires a Reeb graph simplified by all steps: 
\smallskip

\noindent
(1) forget about every node of degree~2 and then merge its two arcs;
\smallskip

\noindent
(2) replace all multiple arcs between the same nodes by a single arc; 
\smallskip

\noindent
(3) remove any subtree (a subgraph without cycles), but keep its root.
\medskip

\noindent
Let $T_k$ be the star graph with $k$ rays.
Then $T_k\times\R$ is a book with $k$ pages.

\begin{thm}
\label{thm:embedding} 
Let a Reeb graph have at most $m$ independent cycles with a common arc.
If a simplified Reeb graph has $n$ nodes, then the original Reeb graph can be canonically embedded into the book $T_{(m+1)\max\{1,n-2\}}\times\R\subset\R^3$.
\end{thm}

The embedding in Theorem~\ref{thm:embedding} is canonical in the sense that different embeddings of the same graph can be continuously deformed to each other.
In a book embedding of a Reeb graph, all nodes lie in the binding axis of the book and each arc lies in a single page. 
%Due to this linear structure, 
A book embedding can be linearly encoded, so we may work with easy codes instead of abstract Reeb graphs.
\medskip

Almost all results about book embedding focus on undirected graphs when nodes can lie in any order in the binding axis.
There is a linear time algorithm to embed any planar graph in 4 pages \cite{Yan86}.
A general undirected graph with $n$ nodes can be embedded in $\lceil\frac{n}{2}\rceil$ pages \cite{CLR87}.
We could find results on book embeddings (stack layouts) of general directed acyclic graphs only for a small number of pages \cite{HPT99}.
The full version of Theorem~\ref{thm:embedding} in section~\ref{sec:directed-graphs} extends \cite[Theorem~2.3]{HPT99}, where any graph with only one cycle is embedded into 2 pages.
In section~\ref{sec:discussion} we review other work on computing Reeb graphs.
\medskip

Our methods for book embeddings were inspired by basic embeddings of graphs \cite{Kur00} and by a reduction of the topological classification for graphs embedded in $\R^3$ to word problems in finitely presented semigroups \cite{Kur07}.
%In more details, any embedded graph $G\subset\R^3$ can be continuously moved to a graph that lies a 3-page book (a union of 3 half-planes) and has all vertices in the binding axis.
%Moreover, all possible isotopies of graphs in $\R^3$ are reduced to finitely many moves between 3-page embeddings of graphs.

%2================================================
%\section{Related work on book embeddings}
%\label{sec:related}

%3================================================
\section{Simplicial complexes and their Reeb graphs}
\label{sec:complexes-Reeb}

\begin{dfn}
\label{dfn:complex}
A \emph{simplicial complex} is a finite set $V$ of vertices 
 and a collection of subsets $\si\subset V$ called \emph{simplices} 
 such that all subsets of a simplex are also simplices. 
The dimension of a simplex $\si=\{v_1,\dots,v_k\}$ is $k-1$.
The dimension of a complex is the highest dimension of all its simplices.
Any simplex inherits the Euclidean topology
 from this geometric realization:
$$\De^d=\{(t_0,\dots,t_{d})\in\R^{d+1}\vl 
 t_0+t_1+\dots+t_{d}=1 \mbox{ and all } t_i\geq 0\}.$$ 
Then we can define the topology on any simplicial complex
 by gluing all its simplices along their common subsimplices,
 so we glue $\si,\tau$ along $\si\cap\tau$.
% as in combinatorial Definition~\ref{dfn:complex}.
\end{dfn}

\begin{dfn}
\label{dfn:Euler-characteristic}
The \emph{Euler characteristic} of a simplicial complex $X$
 is the alternating sum $\chi(X)=\sum\limits_{i\geq 0}c_i(X)$,
 where $c_i(X)$ is the number of $i$-dimensional simplices of $X$.
The sum is finite as $X$ consists of finitely many simplices.
\end{dfn}

\begin{dfn}
\label{dfn:homotopy}
Two maps $f,g:X\to Y$ are \emph{homotopic} if there is a continuous map $\Phi:X\times[0,1]\to Y$, such that $\Phi(x,0)=f(x)$ and $\Phi(x,1)=g(x)$ for any $x\in X$, so $f,g$ are connected by a continuous family of maps $\Phi(\cdot,t):X\to Y$, $t\in[0,1]$.
Two simplicial complexes $X,Y$ have the same \emph{homotopy type}
 (are \emph{homotopy equivalent}) if there are continuous functions
 $f:X\to Y$ and $g:Y\to X$ such that $g\circ f:X\to X$ and $f\circ g:Y\to Y$ are homotopic to the corresponding identity maps $\id_X:X\to X$ and $\id_Y:Y\to Y$.
\end{dfn}

We remind the classical result that $\chi(X)$ is a homotopy invariant of $X$.

\begin{pro}
\label{pro:Euler-characteristic}
If simplicial complexes $X,Y$ have the same homotopy type, then $\chi(X)=\chi(Y)$, so the Euler characteristic is a homotopy invariant.
\end{pro}

%2.2-------------------------------------------------------------------------
%\subsection{Reeb graphs with homology labels on edges}
%\label{sec:Reeb-graphs}

Let $X$ be a simplicial complex
 with a piecewise linear function $f:X\to\R$.
We assume that $f$ is given by its real values at all vertices of $X$.
We consider only generic functions $f$ that have distinct values
 at all vertices. 
Otherwise we symbolically perturb the values of $f$ by replacing the test $f(v_i)\stackrel{?}{<}f(v_j)$ with the test $i\stackrel{?}{<}j$ \cite[p.~253]{PC04}.
Then we linearly extend $f$ to any simplex
  (hence to the whole simplicial complex $X$)
 spanned by vertices $v_0,\dots,v_d$ as follows:
 $f(\sum\limits_{i=0}^d t_i v_i)=\sum\limits_{i=0}^d t_i f(v_i)$,
 where all $t_i\geq 0$ and $\sum\limits_{i=0}^d t_i=1$.

\begin{dfn}
\label{dfn:Reeb-graph}
Let a simplicial complex $X$ have a piecewise-linear function 
 $f:X\to\R$ given by distinct values at all vertices of $X$.
The \emph{Reeb graph} $\Reeb(f,X)$ is the quotient of $X$, 
 where we collapse each connected component of 
 $f^{-1}(t)\subset X$ to a single point.
Different connected components are collapsed to different points.
Any point $p\in\Reeb(f,X)$ represents a connected component
 of $f^{-1}(t)$ and has the associated value $f(f^{-1}(t))=t$.
Then we have the associated function $f_{\Reeb(X)}:\Reeb(f,X)\to\R$.
\smallskip

%For a piecewise-linear function $f:X\to\R$ on a simplicial complex $X$,
A value $f(s)\in\R$ is called \emph{regular} if all level sets $f^{-1}(t)$ have the same number of connected components over a small interval $s-\de<t<s+\de$.
Otherwise the value $f(s)$ is \emph{critical}.
A point in $\Reeb(f,X)$ is a \emph{Reeb node} if the corresponding level set component passes through a critical value of $f$.
\end{dfn}

An arc of $\Reeb(f,X)$ is a pair $(u,v)$ of nodes ordered by their values $f(u)<f(v)$.
If $\Reeb(f,X)$ has $k$ multiple arcs between the same nodes $u,v$, then the pair $(u,v)$ is repeated $k$ times in the full set of arcs.
We call vertices and edges of $\Reeb(f,X)$ \emph{nodes} and \emph{arcs}, respectively, to avoid a confusion with vertices and edges of a simplicial complex.
Later we shall embed all arcs of $\Reeb(f,X)$ as curved (often circular) arcs in a multi-page book.
\smallskip

Definition~\ref{dfn:Reeb-graph} implies that $\Reeb_f(X)$ has only Reeb nodes of degree~1 or at least 3, no nodes of degree~2.
However, if we monitor not only connected components of 
 a level set $f^{-1}(t)$, but also the Euler characteristic,
 then we need to add (topologically trivial) nodes of degree~2
 to $\Reeb_f(X)$ when the Euler characteristic or (more generally) 
 any Betti number of  $f^{-1}(t)$ changes.

\begin{dfn}
\label{dfn:Euler-Reeb}
A point $p\in\Reeb_f(X)$ is called an \emph{Euler node} of degree~2, if the corresponding level set component $C(X,p)$ changes its Euler characteristic when passing through the value $f_{\Reeb(X)}(p)$.
Then every arc in the new \emph{Euler-Reeb graph} $\ER(f,X)$ has the associated Euler characteristic.
\end{dfn}

\begin{figure}[h]
\centering
\includegraphics[scale=1.0]{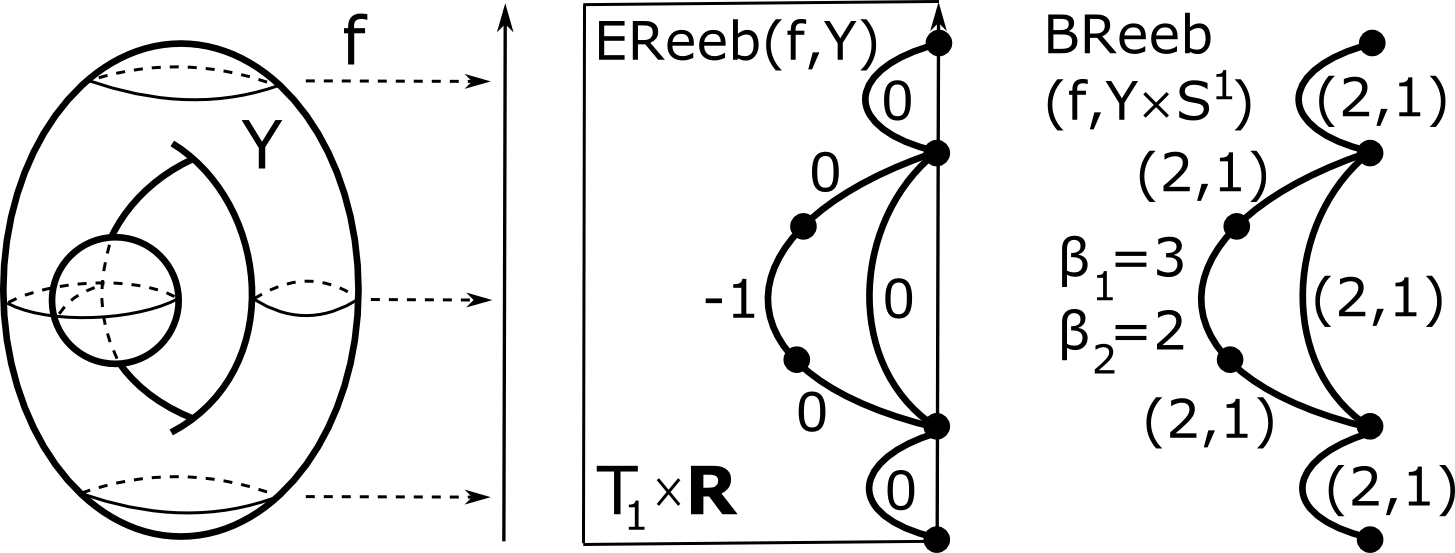}
\caption{
\label{fig:torus-bubble-Reeb}
A torus $Y$ with a bubble, $\ER(f,Y)$ and $\BR(f,Y\times S^1)$.}
\end{figure}

The middle picture in Fig.~\ref{fig:torus-bubble-Reeb} shows $\ER(f,Y)$ for a torus $Y$ with a bubble.
The bubble generates the arc with the associated Euler characteristic $\chi(\theta)=-1$, where the theta-graph $\theta$ is a component in a level set $f^{-1}(t)$.

\begin{dfn}
\label{dfn:Betti-Reeb}
A point $p\in\Reeb_f(X)$ is called a \emph{Betti node} of degree~2, if the corresponding level set component  $C(X,p)$ changes one its Betti numbers $\be_i$, $i\geq 1$, when passing through the value $f_{\Reeb(X)}(p)$.
Each $\be_i$ is the rank of the $i$-th homology group of $C(X,p)$
Every arc in the new Betti-Reeb graph $\BR(f,X)$ has the associated array of Betti numbers $(\be_1,\dots,\be_d)$.
\end{dfn}

The right-hand side picture in Fig.~\ref{fig:torus-bubble-Reeb} shows $\BR(f,Y\times S^1)$.
The bubble in $Y$ gives a component $\theta\times S^1$ (two  tori attached along an annulus) whose Betti numbers are $\be_1=3$ (two meridians, one common longitude) and $\be_2=2$.

%3================================================
\section{Multi-page book embeddings of graphs}
\label{sec:book-embeddings}

\begin{dfn}
\label{dfn:directed-acyclic}
An \emph{undirected cycle} in a graph is a sequence of distinct nodes $p_1,\dots,p_n$ such that any nodes $p_i,p_{i+1}$ are connected by an arc, where $i=2,\dots,n$ and $p_{n+1}=p_1$.
A \emph{directed acyclic graph} (DAG) is a graph with oriented arcs, but without \emph{directed cycles} of consistently oriented arcs. 
\smallskip

Each undirected cycle in $G$ can be considered as a formal sum of its arcs (with coefficients in $\Z/2\Z$ if we forget about orientations).
Several cycles of $G$ are called \emph{independent} if their formal sums are linearly independent in the vector space spanned by all arcs (say with coefficients in $\Z/2\Z$ for simplicity).
\end{dfn}

Embedding Theorem~\ref{thm:embedding} uses the maximum number $m(G)$ of independent undirected cycles that share a common arc.
If cycles of $G$ meet only at nodes or if all cycles are disjoint, they don't share any arcs, so $m(G)=1$.
The graph $\theta_k$ consisting of $k$ arcs $e_1,\dots,e_k$ connecting 2 nodes has $k-1$ independent cycles $e_1\cup e_i$, $i=2,\dots,k$, which share $e_1$, hence $m(\theta_k)=k-1$.

\begin{lem}
\label{lem:Reeb-is-DAG}
Any directed acyclic graph has an ordering of nodes consistent with orientations of arcs, namely any arc is oriented from a smaller node to a larger node.
Any Reeb graph $\Reeb(f,X)$ is a directed acyclic graph.
\end{lem}
\begin{proof}
We can uniquely choose the partial ordering: a node $u$ is smaller than $v$ if there is a path of consistently oriented arcs from $u$ to $v$. Then we may arbitrarily decide which of any two non-comparable nodes is larger.
In $\Reeb(f,X)$ we orient an arc from $u$ to $v$ if $f_{\Reeb(X)}(u)<f_{\Reeb(X)}(v)$.
\end{proof}

\begin{dfn}
\label{dfn:embedding}
A function $f:X\to Y$ between simplicial complexes is called \emph{an embedding} if $f$ is a homeomorphism on image, namely the restriction $f:X\to f(X)$ is bijective and bi-continuous (in both directions).
\end{dfn}

\begin{dfn}
\label{dfn:multi-page-book}
For any integer $k\geq 3$, if we attach $k$ edges (\emph{rays}) to 
 one vertex $c$ (\emph{center}), we get the \emph{star graph} $T_k$. 
The product $T_k\times\R$ is called the \emph{$k$-page book}.
The vertical axis $c\times\R$ is called the \emph{binding axis} (or the \emph{spine}).
For each edge $e\subset T_k$, the product $e\times\R$ is a \emph{page} of $T_k\times\R$, see Fig.~\ref{fig:multi-page-book}.
\end{dfn}

\begin{figure}[h]
\centering
\includegraphics[scale=1.1]{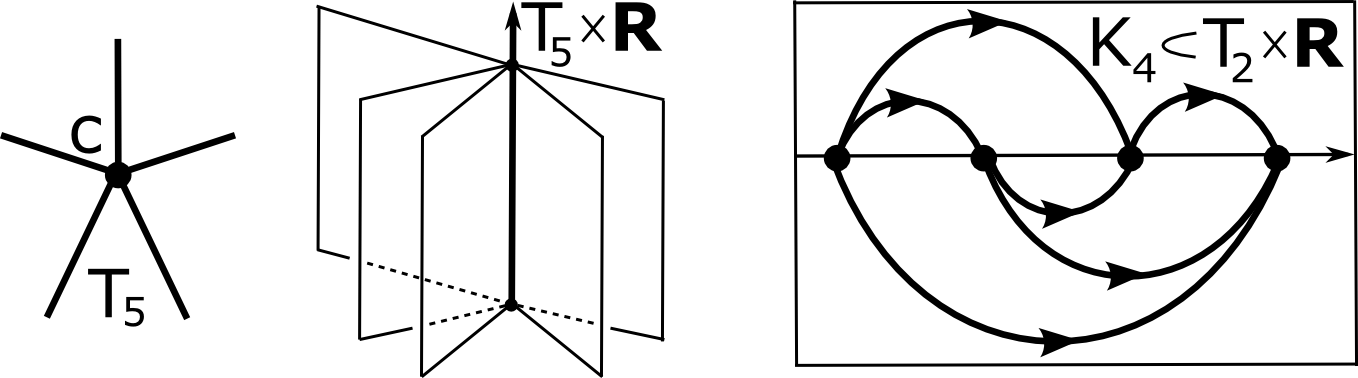}
\caption{
\label{fig:multi-page-book}
The star $T_5$, the book $T_5\times\R$, an embedding $K_4\subset T_2\times\R$.}
\end{figure}

\begin{dfn}
\label{dfn:book-embedding}
An embedding $G\subset T_k\times\R$ of a directed acyclic graph $G$ into the book is called a \emph{book embedding} if the following properties hold:
\smallskip

\noindent
$\bullet$
the \emph{nodes-in-axis} property:
all nodes of $G$ lie in the binding axis $c\times\R$;
\smallskip

\noindent
$\bullet$
the \emph{directional} property:
 any arc has the same direction as the axis $c\times\R$;
\smallskip

\noindent
$\bullet$
the \emph{arc-in-one-page} property:
 any arc is within a single page of $T_k\times\R$.
\end{dfn}

The right hand side picture of Fig.~\ref{fig:multi-page-book} shows a book embedding $K_4\subset T_2\times\R$.
The properties from Definition~\ref{dfn:book-embedding} allow us to encode a book embedding by listing all edges (say, in the increasing order of their lowest node) and by specifying the index of the page containing each arc.
Fig.~\ref{fig:torus-bubble-Reeb} shows more optimal embeddings of $\ER(f,Y)$, $\BR(f,Y\times S^1)$, not book embeddings.

\begin{dfn}
\label{dfn:encoding}
Number nodes of a Reeb graph $G$ by $1,2,\dots,n$.
Then a book embedding $G\subset T_k\times\R$ has the \emph{code} of pairs $(ij)_l$, where $i<j$ are nodes connected by an arc in the $l$-th page.
We may also include labels, e.g. $(ij;\chi)_l$ means that the arc $(ij)$ has the associated Euler characteristic $\chi$.
Pairs of nodes can be lexicographically ordered if we need a unique code.
\end{dfn}

%4================================================
\section{Embeddings of directed acyclic graphs}
\label{sec:directed-graphs}

We were initially motivated by embeddings of Reeb graphs.
However, all results in this section work for more general directed acyclic graphs.

\begin{dfn}
\label{dfn:DAG-tree}
A directed acyclic graph $T$ is called a \emph{tree} if $T$ has no cycles even after forgetting all orientations.
A \emph{root} of $T$ is a node of degree~1.
\end{dfn}

\begin{lem}
\label{lem:embed-tree}
Any directed acyclic tree $T$ with a root $p$ can be embedded into the upper half-disk $\{(x,y)\in\R^2: x^2+y^2\leq\pm\ep x,\; y\geq 0\}$ for any $\ep>0$ in such a way that the root $p$ goes to the origin $(0,0)$, all other nodes lie in the diameter on the $x$-axis and all arcs have the direction of the $x$-axis.
\end{lem}
\begin{proof}
We prove by induction on the number $n$ of nodes. 
The base $n=1$ is trivial.
In the inductive step we may assume that the only arc $e$ at the root $p$ of $T$ is oriented from $p$ to another node $q$.
Then we shall embed $T$ into the positive half-disk with the center $(0.5\ep,0)$ and radius $\ep>0$.
If the arc $e$ is oriented from $q$ to $p$, the symmetric construction gives a required embedding into the negative half-disk with the center $(-0.5\ep,0)$ and radius $\ep$.
\smallskip

We embed the arc $e$ as a curved arc such that the root $p$ goes to the origin $(0,0)$ and the node $q$ goes to the center $(0.5\ep,0)$ of the upper half-disk.
If we remove the arc $e$, the given tree $T$ splits into smaller subtrees with the root $q$.
By the inductive assumption each subtree can be embedded into a small half-disk to the left or to the right of the center $q=(0.5\ep,0)$.
Fig.~\ref{fig:add-trees-to-node} shows yellow half-disks with similarly embedded tripods.
Let the node $q$ have ordered neighbors $q_1<\dots<q_i<q$.
We embed the subtrees of $q_i$ one by one starting from the lowest neighbor $q_1$.
We embed every next subtree into an even smaller half-disk between $q$ and the previous subtree.
We similarly embed subtrees with higher neighbors of $q$ starting from the highest.
\end{proof}

\begin{figure}[h]
\includegraphics[scale=1.1]{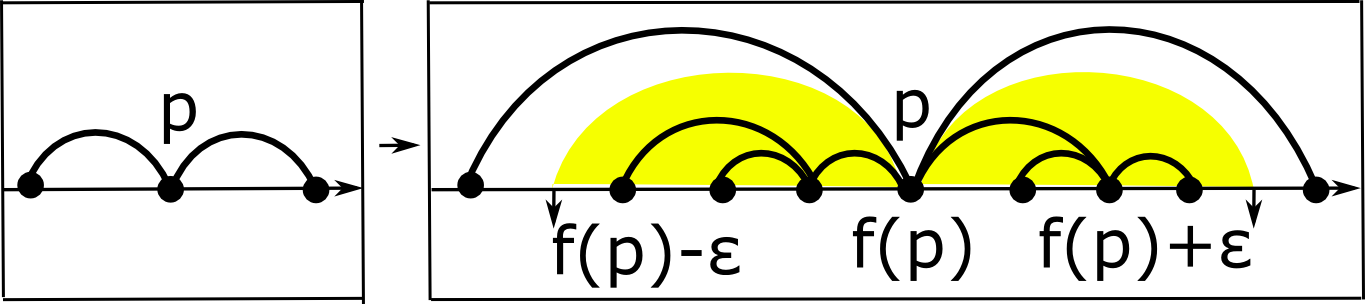}
\caption{
\label{fig:add-trees-to-node}
How to extend a book embedding by adding trees at a node $p$.}
\end{figure}

Any undirected graph $G$ with $n$ nodes has a book embedding into $\lceil\frac{n}{2}\rceil$ pages \cite{CLR87}.
In Lemma~\ref{lem:embed-core} we embed a directed acyclic graph with an extra property that allows us to add later nodes of degree~2 and re-embed any arc in the same page by putting all new nodes into the binding axis $c\times\R$.

\begin{lem}
\label{lem:embed-core}
Any directed acyclic graph $G$ with $n\geq 4$ ordered nodes has a canonical book embedding $G\subset T_{n-2}\times\R$ such that this extra property holds:
\smallskip

\noindent
$\bullet$
the \emph{access-to-axis} property: the projection to $c\times\R$ of any arc contains a small interval not covered by the projections of other arcs in the same page.
\end{lem}
\begin{proof}
We construct a required embedding $G\subset T_{n-2}\times\R$ by induction on the number $n$ of nodes $p_1<\dots<p_n$, which can be ordered by Lemma~\ref{lem:Reeb-is-DAG}.
In the base case $n=4$ we embed $K_4\subset T_2\times\R$
 as shown in Fig.~\ref{fig:multi-page-book}.
In the inductive step from $n$ to $n+1$, we take the $(n+1)$-st highest node $p_{n+1}$ and embed the subgraph of $G$ on $n$ lower nodes by the inductive assumption.
\smallskip

Then we embed all arcs at $p_{n+1}$ into the extra $(n-1)$-st page of the larger book $T_{n-1}\times\R$ as curved arcs connecting $p_{n+1}$ to lower nodes, so we avoid any self-intersections.
The embedding depends only on the order of nodes in the binding axis.
If we deform the nodes keeping their order, then we can continuously deform the embedded graph $G$ within the book $T_{n-2}\times\R$.
\end{proof}

\begin{dfn}
\label{dfn:core-subgraph}
The \emph{core} subgraph $\core(G)$ is obtained from a directed acyclic graph $G$ in several stages where each stage consists of these steps:
\smallskip

\noindent
Step (1): forget about any degree~2 node and merge its two arcs into one; 
\smallskip

\noindent
Step (2): replace all multiple arcs between the same nodes by a single arc; 
\smallskip

\noindent
Step (3): remove any subtree attached to a node $p$ and keep the node $p$. 
\smallskip

%\noindent
At each simplification stage we apply step~(1), then~(2), finally (3) before moving to the next stage when we again apply (1), (2), (3) and so on. 
\smallskip

We have actually introduced the ascending filtration of the subgraphs
$\core(G)\subset\core_1(G)\subset\dots\subset\core_k(G)=G$ when $G$ was simplified to $\core(G)$ through $k$ stages.
Namely, $\core_i(G)$ is obtained from $\core_{i+1}(G)$ in one stage by following Step~(1), then Step~(2), finally Step~(3), $i=1,\dots,k-1$.
\end{dfn}

In fact $\core(G)$ has only nodes of degree at least $3$ and can be considered as a subgraph of $G$.
So the smallest core (except a single point) is $K_4$.
\smallskip

For the Euler-Reeb graph of the torus $Y$ with a bubble in Fig.~\ref{fig:torus-bubble-Reeb} at the first stage we (1) forget about two nodes of degree~2, (2) replace a new double arc by a single arc, (3) simplify the resulting 3-arc path to a single point.
\smallskip

In practice, the Reeb graph $\Reeb(f,X)$ may have few cycles, because any noise in original data usually leads to short edges or trivial nodes of degree~2 that are not included in the core of the graph $\Reeb(f,X)$.
So the core of a Reeb graph is usually small, which leads to embeddings into small books. 

\begin{dfn}
\label{dfn:edge-subgraph}
For any arc $e$ in the subgraph $\core(G)\subset G$ of a directed acyclic graph $G$, we introduce the \emph{edge-subgraph} $G(e)\subset G$ that was collapsed to (replaced by) the arc $e$ in the simplification process from Definition~\ref{dfn:core-subgraph}.
To make $G(e)$ unique, we order all nodes of $G$ by Lemma~\ref{lem:Reeb-is-DAG} and order all arcs with the same endpoints.
Let $T$ be a tree that was attached to a node $p$ and was removed in Step~(3).
We include $T$ into $G(e)$ for the first arc $e$ from the root $p$ of $T$ to the next node $q>p$ in the fixed order of nodes.
\end{dfn}

To clarify Definition~\ref{dfn:edge-subgraph}, we reverse all steps of the last simplification stage from $\core_1(G)$ to $\core(G)$.
If in Step~(3) we removed a subtree $T$ attached to a node $p$, then we include $T$ into the edge-subgraph $G(e)$ for a unique edge at $p$.
If in Step~(2) we replaced $k$ multiple arcs with the same nodes $p,q$ by the single arc $e$, then $G(e)$ contains all these $k$ arcs, so $e$ is blown up $k$ times.
If in Step~(1) we forgot about a degree~2 node in the arc $e$, then we simply add this node to $G(e)$.
After reversing the last simplification stage, we got a larger graph $\core_1(G)\supset\core(G)$.
Now can start building $G(e)$ for any new arc $e\subset\core_1(G)$.
After reversing all stages of Definition~\ref{dfn:edge-subgraph} we have constructed a unique edge-subgraph $G(e)\subset G$ for any arc $e\subset\core(G)$.
\smallskip

A book embedding restricted to a single page is a union of disjoint curved arcs with nodes in the binding axis $c\times\R$. These arcs split the page into intermediate areas, see the yellow area $A(e)$ below the arc $e$ in Fig.~\ref{fig:add-trivial-vertices}.

\begin{figure}[h]
\includegraphics[scale=1.18]{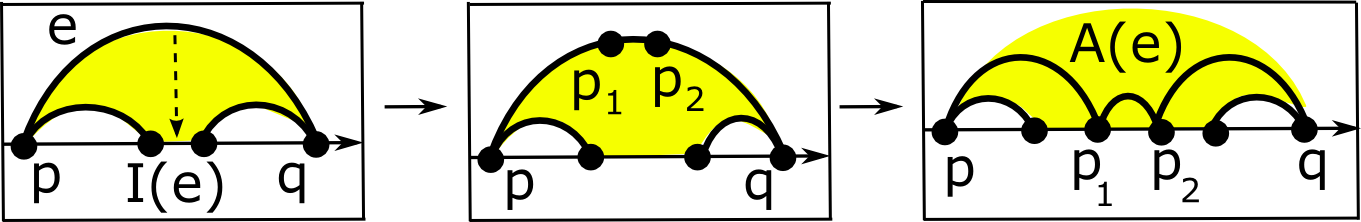}
\caption{
\label{fig:add-trivial-vertices}
How to extend a book embedding after adding nodes of degree~2.}
\end{figure}

\begin{lem}
\label{lem:holed-book}
Let a directed acyclic graph $G$ have at most $m$ independent cycles with a common arc.
Let $\core(G)\subset T_{n-2}\times\R$ be a book embedding from Lemma~\ref{lem:embed-core}.
By the access-to-axis property, for any arc $e\subset\core(G)$, let  $I(e)\subset c\times\R$ be the union of intervals not covered by the projections of other arcs in the same page as $e$. 
Let $A(e)$ be the \emph{open area} bounded by $e$, $I(e)$ and all arcs of $\core(G)$ connecting the endpoints of $e$ and $I(e)$.
Let the \emph{holed} book $B(e)$ consist of $m+1$ copies of $A(e)$ attached along $I(e)$, see the first picture in Fig.~\ref{fig:holed-book-embeddings}.
Then there is a book embedding $G(e)\subset B(e)\subset T_{m+1}\times\R$.
% in the sense of Definition~\ref{dfn:book-embedding}.
\end{lem}
\begin{proof}
Similarly to the discussion after Definition~\ref{dfn:edge-subgraph}, we shall embed $G(e)$ by reversing all simplification stages from Definition~\ref{dfn:core-subgraph}.
Formally we shall argue by induction on the number of stages.
If $\core(G)=G$, then $G(e)=e$ is already embedded.
In the inductive step for the last simplification stage we use embeddings of smaller edge-subgraphs $G(e_i)$ built in fewer stages.
\medskip

\begin{figure}[h]
\includegraphics[scale=1.05]{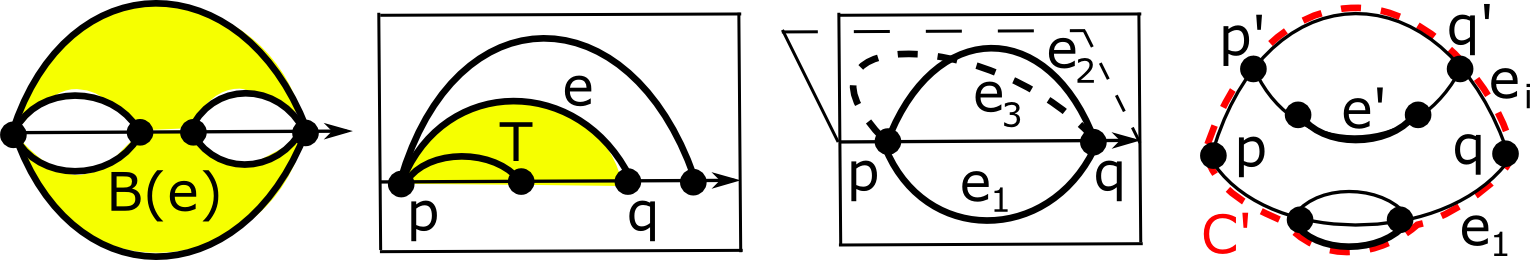}
\caption{
\label{fig:holed-book-embeddings}
A holed 2-page book $B(e)$, cases of Steps (3) and (2) in Lemma~\ref{lem:holed-book}.}
\end{figure}

\noindent
{\bf Case of Step (3)}. 
Let the arc $e\subset\core(G)$ have endpoints $p<q$ and $G(e)$ have a tree $T$ attached at $p$.
By Definition~\ref{dfn:edge-subgraph} the node $q>p$ is the next adjacent node to $p$.
Hence $e$ is the lowest arc (closest to the binding axis $c\times\R$)
 among all arcs at $p$ in the same page.
So the open area $A(e)$ looks like a half-disk bounded by $e$ with the diameter in the binding axis, see the second picture in Fig.~\ref{fig:holed-book-embeddings}.
We embed $T$ into this half-disk by Lemma~\ref{lem:embed-tree}.
\medskip

\noindent
{\bf Case of Step (2)}. 
Let $G(e)$ contain $k$ multiple arcs $e=e_1,\dots,e_k$ that were replaced by the arc $e$ between nodes $p,q$.
These $k$ multiple arcs generate $k-1\leq m$ independent cycles sharing $e=e_1$.
In the simplest case $G(e)=\cup_{i=1}^k e_i$ and we embed each of $k$ arcs $e_i$ one per page (as rotational copies of $e$) into the book $B(e)$ with $m+1\geq k$ pages, see the third picture in Fig.~\ref{fig:holed-book-embeddings}. 
\smallskip

In general, $G(e)$ is the union of all subgraphs $G(e_i)$, $i=1,\dots,k$, built on the arcs $e_1,\dots,e_k$ by reversing earlier simplification stages.
If each arc $e_i$ is shared by at most $m_i$ independent cycles in $G(e_i)$, then we need $m_i$ extra pages to embed $G(e_i)$ by the inductive assumption.
So in total we need $m_1+\dots+m_k$ extra pages for embedding all subgraphs $G(e_i)$.
However, all these $m_1+\dots+m_k$ independent cycles from $G(e)=\cup_{i=1}^k G(e_i)$ give the same number of independent cycles that share the arc $e=e_1$ in $G(e)\subset G$.
\smallskip

Indeed, let each of these $m_i$ cycles contain an arc $e'$, see the last picture in Fig.~\ref{fig:holed-book-embeddings}. 
If $p',q'$ are extreme nodes (according to their order) in such a cycle $C$, we can replace $C$ by the larger cycle $C'$ going through $e=e_1$ (from $p$ to $q$), then along $e_i$ (from $q$ to $p$) following the original cycle $C$ through the nodes $q'$ and $p'$, but avoiding $e'$.
Then all new $m_i$ cycles are independent in $G(e)$ and share the same subarc in $e_1$ as all $m_1$ cycles from $G(e_1)$, $i=2,\dots,k$.
So we have enough $m_1+\dots+m_k\leq m$ pages to embed $G(e)=\cup_{i=1}^k G(e_i)$.
\medskip

\noindent
{\bf Case of Step (1)}. 
Let $G(e)$ have nodes $p_1,\dots,p_j$ of degree~2 on the arc $e$ between nodes $p,q$.
Then we re-embed $e$ in the same page by putting all new nodes into $I(e)\subset c\times\R$.
See Fig.~\ref{fig:add-trivial-vertices} for two nodes $p_1,p_2$.
To make this embedding caninical, we may put all $p_1,\dots,p_j$ into the first interval from (possibly disconnected) $I(e)$ between $p$ and the next node in $c\times\R$.
%By a slight perturbation, we can make sure that the new nodes do not coincide with any other nodes that were previously put into the binding axis $c\times\R$.
\end{proof}
%\bigskip

\noindent
{\bf General version of Embedding Theorem~\ref{thm:embedding}.}
Let a directed acyclic graph $G$ have at most $m$ independent cycles that share a common arc.
If $\core(G)$ has $n$ nodes, then $G$ has a book embedding in $T_{(m+1)\max\{1,n-2\}}\times\R$.
\begin{proof}
If $\core(G)$ is not a single node, then $n\geq 4$ and we take a book embedding $\core(G)\subset T_{n-2}\times\R$ from Lemma~\ref{lem:embed-core}.
If $\core(G)$ is a single node, then we set $n=3$, so $\max\{1,n-2\}=1$ for simplicity.
Then $G$ is the union of the subgraphs $G(e)$ from Definition~\ref{dfn:edge-subgraph} over all arcs $e\subset\core(G)$.
The open areas $A(e)$ from Lemma~\ref{lem:holed-book} are disjoint over all arcs $e\subset G$ in the same page of $T_{n-2}\times\R$.
Hence the union of the holed books $B(e)$ over all $e\subset\core(G)$ from one page of $T_{n-2}\times\R$ are disjoint in the blown-up $(m+1)$-page book.
For all arcs $e\subset\core(G)$ in one page, we can jointly embed $G(e)$ into the disjoint union of the holed books $B(e)$ by Lemma~\ref{lem:holed-book}, hence into the same $m+1$ pages.
So each page of the book $T_{n-2}\times\R\supset\core(G)$ generates $m+1$ more pages and the whole graph $G$ is embedded into $T_{(m+1)(n-2)}\times\R$.
%We exclude the cases of small $m,n$ to guarantee an embedding into the plane $T_2\times\R$.
\end{proof}

If we would like to canonically embed all Reeb graphs into 
 a fixed book, we may choose a star graph with infinitely many edges $T_{\infty}\subset\R^2$ consisting of the radii of the unit circle to the points $\exp\Big(\Big(1-\dfrac{1}{k}\Big)\pi\sqrt{-1}\Big)$, $k\geq 1$.

%6================================================
\section{Discussion and further open problems}
\label{sec:discussion}

Trivial nodes of degree~2 in Reeb graphs $\ER(f,X)$ and $\BR(f,X)$ do not create any problem for encoding book embeddings.
Indeed, we may encode a path between any nodes of degree not equal to 2 by a linear sequence of arcs with all necessary labels.
All extra nodes of degree~2 in Reeb graphs have no attached subtrees.
Then we may leave all multiple arcs between the same nodes in one page.
For instance, the Euler-Reeb graph $\ER(f,Y)$ for a torus $Y$ with a bubble in Fig.~\ref{fig:torus-bubble-Reeb} can be encoded by the array of arcs $(12;0)_1,(23;0)_1,(23;0,-1,0)_1,(34;0)_1$, where $1,2,3,4$ are ordered nodes and the subscript $1$ highlights only 1 page.
One of the arcs between nodes $2$ and $3$ splits into 3 subarcs with the associated Euler characteristics $0,-1,0$.
\smallskip

If we drop the requirement for a book embedding in Definition~\ref{dfn:book-embedding} that all nodes of degree~2 in Reeb graphs lie in the binding axis, then we strengthen Theorem~\ref{thm:embedding} and get optimal embeddings into a 1-page book as in Fig.~\ref{fig:torus-bubble-Reeb}.

\begin{cor}
\label{cor:Reeb-code}
Let a Reeb graph $G$ have at most $m$ independent cycles that share a common arc after the first simplification stage in Definition~\ref{dfn:core-subgraph}.
If $\core(G)$ has $n$ nodes, then $G$ has an embedding into the smaller book $T_{(m+1)\max\{1,n-2\}}\times\R$, which can be linearly encoded as in Definition~\ref{dfn:encoding}.
\end{cor}

The number $m$ of independent cycles in Corollary~\ref{cor:Reeb-code} is smaller for a Reeb graph than in general Theorem~\ref{thm:embedding}, because $G$ may lose multiple arcs after Steps (1) and (2) of the first simplification stage of Definition~\ref{dfn:core-subgraph}.
So the number $m$ cycles is now counted for a graph simpler than $G$.
For instance, the Euler-Reeb graph $G=\ER(f,Y)$ for a torus $Y$ with a bubble in Fig.~\ref{fig:torus-bubble-Reeb} initially has $m=1$ (one cycle).
However $G$ is simplified to a single node in one stage, so $m=0$ in Corollary~\ref{cor:Reeb-code} gives an embedding $G\subset T_1\times\R$. %as in Fig.~\ref{fig:torus-bubble-Reeb}.
\smallskip

The problem to find the minimum page number $\PN(G)$ such that a directed acyclic graph $G$ has a book embedding into $T_{\PN(G)}\times\R$ is NP-complete \cite{HP99}.
Our key result is an algorithmic method to embed any directed acyclic graph $G$ into a small book, which gives a practical upper bound for $\PN(G)$. 
The following problem is the next step for automatic processing Reeb graphs.
%Book embeddings of Reeb graphs respect the order of nodes and do not create any extra intersections with the binding axis.
%Such embeddings are visualized in 3-space and easily encoded.

\begin{prm}
\label{prm:computable-embedding}
Design a fast algorithm to compute a book embedding for a given directed acyclic graph, hence compute a linear code for any Reeb graph.
\end{prm}

%\begin{prm}
%\label{prm:optimal}
%For a directed acyclic graph $G$, find the minimum number $k$ of pages such that $G$ has a book embedding into $T_k\times\R$, see
% Definition~\ref{dfn:book-embedding}.
%\end{prm}

Different shapes may have identical Reeb graphs.
The ultimate aim is to design a shape descriptor that uniquely represents geometric objects up to a natural equivalence.
We start from the homotopy equivalence on complexes.

\begin{prm}
\label{prm:reconstruction}
In addition to Betti numbers on edges of a Reeb graph, find extra invariants that would enable us to reconstruct a simplicial complex $X$ up to (say) a homotopy equivalence from the given graph $\Reeb(f,X)$. 
\end{prm}

It is interesting to translate transformations of time-varying Reeb graphs \cite{EHMPS08} into codes of book embeddings.
The $O(m\log m)$ algorithm \cite{DN12,HWW10,Par12} for computing the classical Reeb graph $\Reeb(f,X)$ can probably be adapted for computing the Euler-Reeb graph.
For 3-dimensional scalar fields \cite{PC04} when level sets are 2-dimensional isosurfaces, all Betti numbers from $\BR(f,X)$ were found without increasing the overall computational complexity. 

\begin{prm}
\label{prm:computation}
For a given simplicial complex with a real-valued function,
 compute the Reeb graph with all extra invariants 
 needed for Problem~\ref{prm:reconstruction}.  
\end{prm}

We are open to collaboration on the above (and any related) problems.
We thank in advance any reviewers for critical comments and suggestions.

%================================================

\end{document}